\documentclass{amsart}
\pdfoutput=1
\usepackage{amsmath}
\usepackage{graphicx}
\usepackage{placeins}
\usepackage{amsthm}
\usepackage{nicefrac}
\usepackage{bbm}
\usepackage{wrapfig}
\usepackage{amssymb}
\usepackage{pgf,tikz}
\usepackage{mathrsfs}
\usetikzlibrary{arrows,decorations.pathreplacing,patterns}
\usepackage{fullpage}
\usepackage{color}
\usepackage{makecell}
\usepackage{pifont}
\newtheorem{thm}{Theorem}
\newtheorem{cor}{Corollary}
\newtheorem{lem}{Lemma}
\theoremstyle{definition}
\newtheorem{defn}{Definition}
\DeclareMathOperator{\tr}{Tr}

\newcommand{\ket}[1]{{\left\vert{#1}\right\rangle}}
\newcommand{\bra}[1]{{\left\langle{#1}\right\vert}}
\newcommand{\braket}[2]{{\left< {#1} \middle\vert {#2}\right>}}
\newcommand{\ketbra}[1]{{\left\vert {#1}\middle\rangle\middle\langle{#1}\right\vert}}
\newcommand{\sprod}[2]{\left|\left< {#1} \middle| {#2} \right>\right|}
\newcommand*\dif{\mathop{}\!\mathrm{d}}
\makeatletter
\renewcommand*{\@textcolor}[3]{%
  \protect\leavevmode
  \begingroup
    \color#1{#2}#3%
  \endgroup
}
\makeatother

\begin{document}
\title{How (maximally) contextual is quantum mechanics?}
\author{Andrew W. Simmons}
\address{Department of Physics, Imperial College London, SW7 2AZ.}
\email{andrew@simmons.co.uk}
\begin{abstract}
Proofs of Bell-Kochen-Specker contextuality demonstrate that there exists sets of projectors that cannot each be assigned either 0 or 1 such that each basis formed from them contains exactly one 1-assigned projector. Instead, at least some of the projectors must have a valuation that depends on the \emph{context} in which they are measured. This motivates the question of \emph{how many} of the projectors must have contextual valuations. In this paper, we demonstrate a bound on what fraction of rank-1 projective measurements on a quantum system must be considered to have context-dependent valuations as a function of the quantum dimension, and show that quantum mechanics is not as contextual, by this metric, as other possible physical theories. Attempts to find quantum mechanical scenarios that yield a high value of this figure-of-merit can be thought of as generalisations or extensions of the search for small Kochen-Specker sets. We also extend this result to projector-valued-measures with projectors of arbitrary equal rank.
\end{abstract}
\maketitle
\section{Introduction}

The Bell-Kochen-Specker theorem tells us that there exist sets of projectors onto states $\{\ket{\psi_i}\}$, such that it is impossible to assign to each of them a valuation of 0 or 1 in a context-independent way, and have there be exactly one 1-valued projector in every PVM context. We might ask the following question: under an assumption of outcome-definiteness, how many such projectors \emph{can} we give a noncontextual valuation to, and how many projectors (or more robustly, what \emph{fraction} of projectors), must be considered to have context-dependent valuations? We will be looking for an assignment to all projection operators onto a Kochen-Specker set $\{\ketbra{\psi_i}\}\rightarrow \{0,C,1\}$ such that the following rules hold: in any basis formed from the projectors onto $\{\ketbra{\psi_i}\}$, there must be at least one with a $C$- or $1$-valuation, and no more than one with a $1$-valuation. The fraction which must be given a $C$-valuation is the key figure-of-merit that will be explored in this paper. 

Historically, Kochen-Specker style proofs of contextuality, also known as strong or maximal contextuality, have tended to have a ``lynch pin" quality; removal of any part of their structure causes the proof to fall apart in its entirety. In fact, this fact has been exploited in order to tame the maximal contextuality demonstrated by the Peres-Mermin magic square for all qubit states \cite{Berm2016}, as we need only remove the ability to measure one context before it fails to be a proof of contextuality at all. For such proofs, then, since removing one context is sufficient to remove all contextuality, certainly allowing a projector to vary in its value assignment by context  (by assigning it a $C$) will. This minimalism is partially motivated by the fact that since the original paper by Kochen and Specker, there has been interest in trying to find examples of small Kochen-Specker sets \cite{Aren2011}. When only a single projector need be considered to take a contextual valuation, finding a small such set is the only way of attaining a higher fraction of contextual projectors. The figure-of-merit introduced above, then, can be seen as a generalisation or successor to the search for small Kochen-Specker sets.

In this paper, we will use graph-theoretic methods as a computational tool to investigate how robust, in this sense, quantum-mechanically accessible Kochen-Specker proofs can be. We will demonstrate that it is possible that other physical theories can be more contextual, given this definition, than quantum mechanics. We can consider this as a stronger analogue than the observation that a PR-box is more nonlocal with respect to the CHSH inequality than any quantum-mechanically accessible scenario.

One metric for contextuality is the \emph{contextual fraction} \cite{Abra2017}, in which we consider our probability distribution to be a convex mixture between a noncontextual theory and a contextual one. This notion is what prompts the term ``maximal contextuality''; if a scenario is maximally contextual then it has a contextual fraction that is the highest possible, \emph{viz} 1. It is independent under a  re-labelling the measurement effects, but the contextual fraction is hard to calculate when the number of corners of the noncontextual polytope is large, which is true in many natural situations such as in nonlocality scenarios in which one party has access to measurements with three or more outcomes \cite{SimmCC}. The figure-of-merit explored in this paper acts as a complementary measure of contextuality that can distinguish different scenarios which have the maximum possible contextual fraction of 1. However, it is likely to be hard to calculate in many circumstances, and in the case in which one is using the entirety of Hilbert space as one's vector set, we shall explore the connection to open problems in mathematics known as Witsenhausen's problem \cite{Wits1974} and the Hadwiger-Nelson problem \cite{Soif2008}.


\section{How robust can quantum maximal contextuality be?}\label{sec:main}

Consider the orthogonality graph $\mathcal{G}$ of a maximal nonlocality scenario, in which we identify projectors with vertices, and two vertices $g_1$, $g_2$, identified with $\ket{\psi_1}$, $\ket{\psi_2}$ are joined by an edge iff $\braket{\psi_1}{\psi_2}=0$. Following Renner and Wolf \cite{Renn2004}, we will restrict our attention to outcome assignments in which we do not ever assign a ``1'' valuation to two orthogonal vectors; this is known as a \emph{weak Kochen-Specker set}. A weak Kochen-Specker set can be extended to a full Kochen-Specker set with the addition of $O(|\ket{\psi_i}|^2d)$ extra vectors, or we can consider the incomplete bases as being completed by additional projectors, whose valuations are allowed to be contextual as they are outside the scope of our scenario. In fact, this distinction will not be crucial for any of the main results; it does however make certain definitions conceptually clearer. An \emph{independent set} is a subset of vertices, none of which are joined by an edge. Any noncontextual assignment of ``1''s, then, will form an independent set since we cannot have two orthogonal vectors, (which correspond to two vertices joined by an edge) in our set of vectors which are noncontextually given ``1'' assignments. The size of the largest independent set of a graph $\mathcal{G}$ is denoted $\alpha(\mathcal{G})$ and is called the \emph{independence number}. Another important concept will be that of a clique:

\begin{defn}[Clique]
For a graph $\mathcal{G}=\{V,E\}$, vertex subset $C\subset V$ is a \emph{Clique} if the subgraph induced from $\mathcal{G}$ on $C$ is \emph{complete}, that is, every vertex in $C$ is connected to every other; and $C$ is inclusion-maximal.
\end{defn}

We say that a clique is a \emph{maximum clique} if there is no clique in $\mathcal{G}$ with more vertices. The size of a maximum clique in $\mathcal{G}$ is denoted $\omega(\mathcal{G})$, and so we note that for a $d$-dimensional quantum system, as long as the set of vectors we can measure includes at least one basis, we will have $\omega(\mathcal{G})=d$. In a hidden variable model for a scenario, we require that in every context, there be at least one vector that can take a ``1'' value either contextually or noncontextually. This motivates the following definition:

\begin{defn}[Maximum-clique hitting-set]
For a graph $\mathcal{G}=\{V,E\}$, a vertex subset $T\subset V$ is a \emph{maximum-clique hitting-set} or \emph{maximum-clique transversal} if for all cliques $C$, $|C|=\omega(\mathcal{G})$, $T\cap C\neq\emptyset$.
\end{defn}

In any contextual hidden-variable model, the set of vertices associated with a set of vectors given either noncontextual or contextual ``1'' assignments is a Maximum-clique hitting set, since every context needs an outcome. This allows us to define our first figure-of-merit for a set of vectors forming a Kochen-Specker type proof of the Bell-Kochen-Specker theorem.

\begin{equation}q_s(\mathcal{G})=\min_T \min_{A\subset T} |T|-|A|,\end{equation}
where $T$ is a maximum-clique hitting-set. We see that this is a graph-theoretic representation of the figure-of-merit we posited above: we take a set $T$ of vectors which will recieve a ``1'' valuation, either contextually or noncontextually, and subtract from it the number of those vectors that can support a noncontextual valuation. This quantity then represents the number of vectors that must receive a contextual valuation in any model for the scenario. However, we can note that it is possible to achieve an arbitrarily high value of $q_s$ just by taking an equal number of noninteracting copies of, for example, the Peres-Mermin magic square. In order that we might more directly compare this quantity between graphs of different sizes, we can consider its graph density:
\begin{equation}
q(\mathcal{G})=\frac{q_s(\mathcal{G})}{|V(\mathcal{G})|}.
\end{equation}
We now have a graph-theoretic formulation of our quantity of interest. What values of this quantity are obtainable, either by quantum-mechanically realisable scenarios, or in general probabilistic theories? We note the following trivial lower bound:
\begin{equation}
q_s(\mathcal{G})\geq \min |T|-\alpha(\mathcal{G}).
\end{equation}
We note that in general, it is \textbf{NP}-complete to calculate the value of $\min |T|$ if we know $\omega(\mathcal{G})$, since for triangle-free graphs, this is equivalent to finding a vertex cover, which is \textbf{NP}-complete \cite{Polj1974}. Also, it is \textbf{NP}-complete to calculate $\alpha(\mathcal{G})$ \cite{Gare1979}. This suggests that calculation of this quantity may be computationally difficult. A tension exists in trying to optimise $q(\mathcal{G})$. Particularly well-connected graphs may achieve an exponentially small independence density-- this is true for a few of the families of graphs considered in this paper-- so few vectors can receive a noncontextual assignment of 1, but this same connectedness means that allowing a single vector's valuation to vary contextually ``fixes'' many contexts.

\begin{figure}
\includegraphics[width=\textwidth]{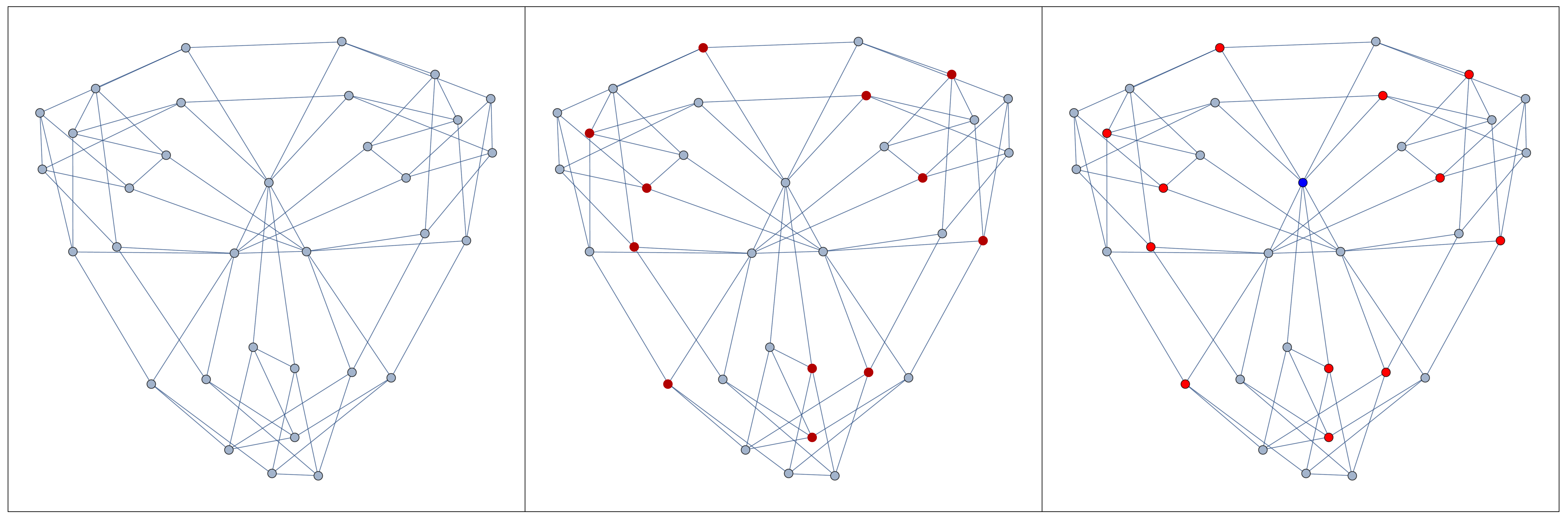}
\caption{Three images illustrating the calculation of the figure of merit. The first is the orthogonality graph for the 33-ray proof of the Kochen-Specker theorem by Peres \cite{Pere1991}. Highlighted in the second in red is a maximal independent set, that represents a maximal assignment of noncontextual 1-valuations. In fact, only a single context lacks a 1-valued vector, visible in the centre of the graph. Highlighted in the third in blue is a projector that can be given a contextual valuation in order that every context contain a ray assigned either a 1 or a contextual valuation; it will take a 1 valuation in the ``central'' context, and a 0 in all others.}
\end{figure}

We will now demonstrate an upper bound on the quantity $q(\mathcal{G})$ for any quantum-mechanically accessible $\mathcal{G}$ in terms of the quantum dimension of the vectors that give rise to it.
\begin{thm}\label{main}
Any orthogonality graph $\mathcal{G}$ induced by a set of vectors $\ket{\psi_i}$ in a $d$-dimensional Hilbert space has
\begin{equation}
q(\mathcal{G}) \leq \left(1-\frac1d\right)^{d-1}-\frac{1}{2^{d-1}}.
\end{equation}
\end{thm}
\begin{proof}
The proof proceeds by demonstrating an explicit procedure for finding a $T$ and an $A\subset T$, for any $\{\ket{\psi_i}\}$, with upper-bounded worst-case performance.

The set of vectors $\ket{\phi}\in\mathcal{H}_d$ such that $\sprod{\psi}{\phi}^2\geq t$ for some $t\in\mathbb{R}$ forms a closed complex hyperspherical cap centred at $\ket{\psi}$. Consider a basis for $\mathcal{H}_d$. Without loss of generality, we can take this to be the standard basis for the set. By symmetry concerns, the quantity $\min_i \left|\braket{\psi}{i}\right|$ is maximised (albeit nonuniquely) by $\ket{\psi}=d^{\nicefrac{-1}{2}}\sum_i\ket{i}$. This puts an upper bound on the size of the largest hyperspherical angle that can separate an arbitrary vector from its closest basis element. A closed complex hyperspherical cap subtending this angle, $\left\{\ket{\psi}\middle|\sprod{\psi}{\phi}^2\geq \nicefrac{1}{d}\right\}$, will then necessarily capture at least one element from the basis, and so any set of vectors within such a hyperspherical cap $C$ construes a maximum-clique hitting-set, $T_C$.

Since columns of a Haar-random unitary matrix over $\mathbb{C}^d$ are uniformly random unit vectors, the volume $V$ of this hyperspherical cap as a proportion of that of the sphere is given by $1-F(t)$, where $F(t)$ is the cumulative distribution function of the squared  modulus of an entry in such a Haar-random unitary. This quantity is well known in the field of random matrix theory and it can be found, for example, as a special case of a result by \.{Z}yczkowski and Sommers \cite{Zycz2000}, and will be demonstrated in Section \ref{highrank}.

\begin{equation}
\frac{V}{V_S}=\left(1-t \right)^{d-1}.
\end{equation}
Clearly, the value of $q(\mathcal{G})$ for a quantum graph is upper bounded by the size of any maximum-clique hitting set, so if we can bound the minimum number of vectors we must intersect with such a hyperspherical cap, this also constitutes a bound on the value $q(\mathcal{G})$. However, we would not expect this bound to be tight, since we have ignored the independent set subtrahend entirely. We would like to demonstrate an independent set within our maximum-clique hitting-set $T_C$, and therefore give a better upper bound on the quantity $q_s(\mathcal{G})=\min_T \min_{A\subset T} |T|-|A|$, giving an explicit construction procedure for a $T$ and $A\subset T$.

We note that a hyperspherical cap containing the vectors $\ket{\phi}$ such that $\sprod{\psi}{\phi}^2> \nicefrac12$ cannot contain two orthogonal vectors as a simple application of the triangle inequality. We can place this smaller hyperspherical cap anywhere within the larger hyperspherical cap, and this will form an independent set $A$ that is a subset of $T_C$. For the rest of this proof, we will assume that the two hyperspherical caps are centred on the same point; this will make no difference to the worst case scenario which forms the bound. The region of Hilbert space, then, in which vectors count towards the figure of merit, is a member of a two-parameter family of \emph{annuli} on the complex hypersphere, illustrated in Figure \ref{annulus} and defined by
\begin{equation}
t_1\leq\sprod{\psi}{\phi}^2\leq t_2.
\end{equation}
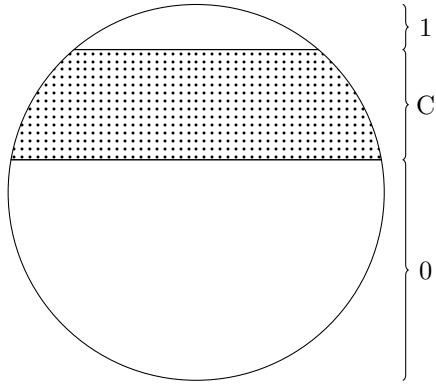
\begin{figure}
\begin{tikzpicture}[line cap=round,line join=round,>=triangle 45,x=2.5cm,y=2.5cm]
\clip(-1.2,-1.2) rectangle (1.4,1.2);
\draw(0.,0.) circle (2.5cm);
\draw (-0.98408,0.173648)-- (0.98408,0.17365);
\draw (-0.649967502437297,0.7599620028497625)-- (0.6499675024372968,0.7599620028497626);
\draw [decorate,decoration={brace,amplitude=2pt,mirror},xshift=0pt,yshift=0pt] (1.1,0.7599620028497626) -- (1.1,1) node [black,midway,xshift=0.3cm,yshift=0.0cm] {1} ;
\draw [decorate,decoration={brace,amplitude=2pt,mirror},xshift=0pt,yshift=0pt] (1.1,0.173648) -- (1.1,0.7599620028497626) node [black,midway,xshift=0.3cm,yshift=0.0cm] {C} ;
\draw [decorate,decoration={brace,amplitude=2pt,mirror},xshift=0pt,yshift=0pt] (1.1,-1) -- (1.1,0.173648) node [black,midway,xshift=0.3cm,yshift=0.0cm] {0} ;
\path[clip] (0,0) circle (1);
\fill[draw=black, pattern=dots] (-0.98408,0.173648)--(-0.98408,0.7599620028497626)--(0.98408,0.7599620028497626)--(0.98408,0.173648) -- cycle;
\end{tikzpicture}
\caption{An illustration of how this construction splits up the Hilbert space into different regions, that can be associated with projectors given noncontextual 1-assignments, contextual projectors, and projectors given noncontextual 0-assignments.}
\label{annulus}
\end{figure}
Such an annulus takes up a proportion 
\begin{equation}p(t_1,t_2)=\left(1-t_1 \right)^{d-1}-\left(1-t_2 \right)^{d-1}\end{equation}
of the Hilbert space. Our construction has $t_1=\nicefrac1d$, and $t_2=\nicefrac12$, giving us a hyperspherical annulus defined by
$\nicefrac1d\leq\sprod{\psi}{\phi}^2\leq \nicefrac12$.
It now remains to be proved that no matter the positions of the set of vectors $\{\psi\}$, there must be some place in which we can place this annulus such that we only intersect strictly less than $p(t_1,t_2)$ of them. 

Consider the canonical action of $SU(d)$ on the unit hypersphere in $\mathbb{C}^d$, which we shall denote $S_\mathbb{C}^{d-1}$. We define a space $SU(d)\times S_\mathbb{C}^{d-1}$, where $SU(d)$ is equipped with the Haar measure, and $S_\mathbb{C}^{d-1}$ is equipped with the uniform measure. Equivalently, we could have taken $H_d$ rather than $S_\mathbb{C}^{d-1}$, in which case it would be equipped with the Fubini-Study metric.

At each $\ket{\psi_i}$, we place a bump function $f^\epsilon_\ket{\psi_i}$ with unit integral and support only within a disc of radius $\epsilon$ around $\ket{\psi_i}$.
Let $A_{t_1,t_2}$ denote the hyperspherical annulus $\{\ket{\phi}\mid t_1\leq\sprod{0}{\phi}^2\leq t_2\}$. By Fubini's theorem, we have:
\begin{equation}
\int_{SU(d)} \int_{S_\mathbb{C}^{d-1}} \mathbbm{1}_{g.z\in A_{t_1,t_2}} \sum_i f^\epsilon_\ket{\psi_i}(z) \dif z \dif g = \int_{S_\mathbb{C}^{d-1}} \int_{SU(d)}  \mathbbm{1}_{g.z\in A_{t_1,t_2}} \sum_i f^\epsilon_\ket{\psi_i}(z) \dif g \dif z
\end{equation}


We note that if we extend our annulus's extent by $\epsilon$, so that we have $t_1\rightarrow t_1-\epsilon$ and $t_2\rightarrow t_2+\epsilon$, we capture the entirety of the measure of the bump functions centred inside the original annulus. Hence, for all $g\in SU(d)$ we have $\left|\left\{\ket{\psi_i} \middle| \ket{\psi_i} \in g^{-1}A_{t_1,t_2}\right\}\right|\leq \int_{S_\mathbb{C}^{d-1}} \mathbbm{1}_{g.z\in A_{t_1-\epsilon,t_2+\epsilon}} \sum_i f^\epsilon_\ket{\psi_i}(z) \dif z $, and so we have
\begin{equation}
\int_{SU(d)} \left|\left\{\ket{\psi_i} \middle| \ket{\psi_i} \in g^{-1}A_{t_1,t_2}\right\}\right| \dif g \leq \int_{S_\mathbb{C}^{d-1}} \int_{SU(d)}  \mathbbm{1}_{g.z\in A_{t_1-\epsilon,t_2+\epsilon}} \sum_i f^\epsilon_\ket{\psi_i}(z) \dif g \dif z,
\end{equation}

Rearranging:
\begin{equation}
\int_{SU(d)} \left|\left\{\ket{\psi_i} \middle| \ket{\psi_i} \in g^{-1}A_{t_1,t_2}\right\}\right| \dif g \leq \sum_i\int_{S_\mathbb{C}^{d-1}}f^\epsilon_\ket{\psi_i}(z)  \int_{SU(d)} \mathbbm{1}_{g.z\in A_{t_1-\epsilon,t_2+\epsilon}} \dif g \dif z \label{eqn:1} 
\end{equation}

Twirling the indicator function $ \mathbbm{1}_{g.z\in A_{t_1-\epsilon,t_2+\epsilon}}$ with respect to the Haar measure over $SU(n)$ reduces it to a constant function with value $p(t_1-\epsilon,t_2+\epsilon)$, the proportion of the Hilbert space taken up by the annulus.

\begin{align}
\int_{SU(d)} \left|\left\{\ket{\psi_i} \middle| \ket{\psi_i} \in g^{-1}A_{t_1,t_2}\right\}\right| \dif g &\leq p(t_1-\epsilon,t_2+\epsilon)\sum_i\int_{S_\mathbb{C}^{d-1}}f^\epsilon_\ket{\psi_i}(z)  \dif z \label{eqn:2}\\
&\leq p(t_1-\epsilon,t_2+\epsilon)\left|\ket{\psi_i}\right|
\end{align}
Since the left hand side of Equation \ref{eqn:2} represents an average over the rotation group, there must be some $g\in SU(d)$ such that $\left|\left\{\ket{\psi_i} \middle| \ket{\psi_i} \in g^{-1}A_{t_1,t_2}\right\}\right|\leq p(t_1-\epsilon,t_2+\epsilon)\left|\ket{\psi_i}\right|$, and this is true for all $\epsilon$. Therefore:
\begin{align}
q(\mathcal{G})&\leq \lim_{\epsilon\rightarrow0}p(t_1-\epsilon,t_2+\epsilon) \\
&\leq \left(1-\frac1d\right)^{d-1}-\frac{1}{2^{d-1}}.
\end{align}

\end{proof} 
\begin{cor}
Any quantum mechanically accessible $\mathcal{G}$ has  $q(\mathcal{G})\leq\nicefrac{4251920575}{11019960576}\sim0.385838$.
\end{cor}
\begin{proof}
Calculating the derivative of $ \left(1-\frac1d\right)^{d-1}-\frac{1}{2^{d-1}}$ with respect to $d$ shows that it takes a maximum between $d=9$ and $d=10$. Evaluation of the quantity at these points reveal the maximum to be at $d=9$, when we get $\left(1-\nicefrac19\right)^{8}-\nicefrac{1}{2^{8}}=\nicefrac{4251920575}{11019960576}$. This, then, forms a hard limit of $q(\mathcal{G})$ for any quantum mechanically achievable $\mathcal{G}$.
\end{proof}
We note that as the quantum dimension approaches infinity, we achieve a limiting value of $\nicefrac1e$, so quantum systems with very high dimension are viable candidates for providing robust contextuality scenarios. In Figure \ref{boundgraph}, we see the bound on $q(\mathcal{G})$ plotted as a function of $d$.

\begin{figure}
\begin{center}
\includegraphics[width=0.5\textwidth]{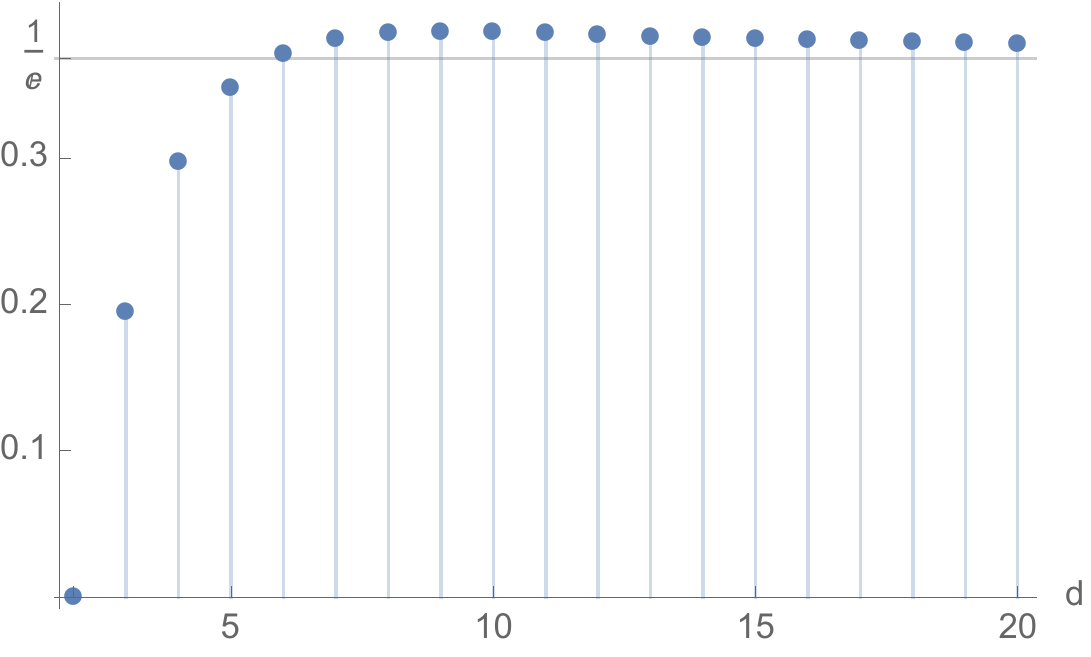}
\caption{A chart showing the value of the bound on $q(\mathcal{G})$ for quantum-mechanically accessible $\mathcal{G}$ as a function of the Hilbert space dimension $d$.}
\label{boundgraph}
\end{center}
\end{figure}
Above, it was shown that a spherical cap taking up a proportion of $(1-\nicefrac1d)^{d-1}$ of Hilbert space must capture at least one vector from every orthonormal basis. However, if for many bases this spherical cap captures more than one basis element, this might suggest that this spherical cap was not the most efficient way of choosing a vector from each basis. We note, though, that the exact geometry of the set of projectors will have a large effect on the efficiency of this construction. One one hand, an increase in the quantum dimension allows for more complicated structures and interplay of the permitted projectors that, as we shall demonstrate, seems to allow for construction of quantum graphs with very small independence number. However, it also has a dampening affect on the quantity $|T|$ since the bases themselves have more geometrical constraints. 

In fact such an annulus can in the worst case can contain every vector from a basis. Below, we will see a specific example of a set of vectors providing a nonlocality scenario, in which \emph{every} projector is captured by a pessimally-chosen spherical cap position. In $d$-dimensional Hilbert space, a spherical cap of proportion $(1-\nicefrac1n)^{d-1}$ can capture $n$ vectors out of a basis. So for any chosen $n'$, as $d$ increases, we see that the fraction of the spherical cap in the construction that cannot capture $n'$ of a basis tends to 0. However, if we consider the fraction of the spherical cap that can capture a constant proportion of basis elements $c$, we get the following behaviour under increasing $d$:
\begin{equation}
\lim_{d\rightarrow\infty} \left(1-\frac{1}{cd}\right)^{d-1}=e^{-\frac1c}.
\end{equation}
This might imply that the quality of the bound does not degrade as the dimension increases.


\begin{figure}
\begin{center}
\begin{tabular}{|c | c | c| c|c| }\hline
Name &  \makecell{Quantum \\Dimension}&\makecell{Transversal\\ Size} & \makecell{Independence \\Number} & \makecell{Figure of \\Merit}\\ \hline
\makecell{ABK Triangle-Free\\Graphs \cite{Alon2010}} &\ding{55}& $n-c\sqrt{n\log n}$ & $c\sqrt{n\log n}$  & $1-O\left(\frac{\log n}{n}\right) $\\\hline
\makecell{Quantum Upper\\ Bound} &$d$& $\leq n\left(1-\frac1d\right)^{d-1}$ & $\geq\frac{n}{2^{d-1}}$&$\leq 0.385838$ \\\hline
\makecell{$E_8$ root \\ system} &8&$\leq \frac{11n}{60}$& $\frac{n}{15}$& $\leq0.11\dot{6}$ \\ \hline
\makecell{Two-Qubit Stabiliser \\ Quantum Mechanics} &4&  $\frac{3n}{10}$& $ \frac{n}{5}$&0.1\\\hline
\makecell{Peres-Mermin \\ Magic Square \cite{Pere1991}} &4& $ \frac{7n}{24}$& $\frac{5n}{24}$ &$0.08\dot{3}$\\\hline
\makecell{Cabello's 18-ray \\  proof \cite{Cabe1997}} &4&$\frac{5n}{18}$&$\frac{4n}{18}$& $0.0\dot{5}$ \\\hline
\makecell{Peres's 33-ray \\  proof \cite{Pere1991}} &3&$\frac{9n}{33}$&$\frac{12n}{33}$& $0.\dot{0}\dot{3}$ \\\hline

\end{tabular}
\end{center}
\caption{A table comparing the graph constructions in this paper; each concrete example represents the graph union of multiple noninteracting copies of the KS proof. Two-qubit stabiliser quantum mechanics forms a KS proof using 60 rays and 105 contexts; the $E_8$ root system forms a KS proof using 120 rays and 2025 contexts. The quantities related to each were calculated using a proof by exhaustion, or are the best found during a non-exhaustive search, such as is the case for the $E_8$ system. It is perhaps of note that the lower bound of the size of a traversal less the independence number forms a trivial bound only for Peres's 33-ray proof, and in fact is optimal for Cabello's 18-ray proof and the Peres-Mermin magic square.}
\end{figure}
\subsection{Triangle-Free Graphs}

In triangle-free graphs, the maximum cliques are merely edges, and so each context consists of two measurement objects. Quantum mechanically, the only such graphs are uninteresting: the disjoint graph union of $K_2$ graphs. These do not have enough structure to form a proof of the Bell-Kochen-Specker theorem; this reflects the fact that each projector appears in only a single context. Many such graphs, then, are not quantum-mechanically accessible. However, they do allow us to prove a concrete bound on what values of $q(\mathcal{G})$ are achievable within the framework of generalised probabilistic theories. In such graphs, the concept of a maximum-clique hitting-set reduces to that of a covering of edges by vertices, known as a \emph{vertex cover}. The size of the minimal vertex cover of $\mathcal{G}$ is denoted $\tau(\mathcal{G})$. This reduction of our figure-of-merit allows us to make use of a bound derived for Ramsey type problems.

In the triangle-free case, we have a powerful relation between maximum-clique hitting-sets and independent sets; namely that they are graph complements of each other. A maximum-clique hitting-set, or vertex cover, is a set of vertices $T$ such that for every edge, at least one of that edge's vertices is in $T$. We can see, then, that the set $V-T$ must be independent, since if there were two vertices in $V-T$ connected by an edge, then neither of those vertices would be in $T$, a contradiction. Conversely, if we have an independent set $A$, then $V-A$ must be a vertex cover; if there were an edge with neither vertex in $V-A$, then both vertices are in $A$, a contradiction. This means that we can bound $q_s(\mathcal{G})$ as
\begin{equation}
q_s(\mathcal{G})\geq \min_\mathcal{G} |T| - \alpha(\mathcal{G}) = |V(\mathcal{G})|-2\alpha(\mathcal{G}).
\end{equation}
In other words, to be able to lower bound $q_s(\mathcal{G})$, we need only consider the independence number, $\alpha(\mathcal{G})$. We seek, therefore, triangle-free graphs with low independence number. We can invoke here a theorem due to Alon, Ben-Shimon and Krivelevich \cite{Alon2010}.
\begin{thm}
There exists a constant $c$ such that for all $n\in\mathbb{N}$ there exists a regular triangle-free graph $\mathcal{G}_n$, with $V(\mathcal{G}_n)=n$, and $\alpha(\mathcal{G}_n)\leq c\sqrt{n\log n}$.
\end{thm}

Hence, we have
\begin{align}
\lim_{n\rightarrow\infty}q(\mathcal{G}_n)&\geq\lim_{n\rightarrow\infty}\frac{n-2c\sqrt{n\log n}}{n}\\
&\geq 1-2c\lim_{n\rightarrow\infty}\sqrt{\frac{\log n}{n}}=1.
\end{align}
Therefore, $q(\mathcal{G}_n)=1$ is approachable in the limit of $n\rightarrow\infty$. We see then that quantum mechanics does not display maximal contextuality as robustly, given this definition of robustness, as other possible GPTs. This could potentially have an impact on attempts to recreate quantum mechanics as a principle theory.

\FloatBarrier

\section{Higher-rank PVMs}\label{highrank}

We now extend the result to apply to projective valued measures including projectors of rank greater than 1, although we are still considering the case in which contexts are made up of a set of projectors which sum to the identity, rather than the more general case of a set of commuting projectors. We will allow a rank-$k$ projector $P$ to ``inherit'' a labelling as a function of the assignments to some $\left\{\ket{\psi_P^{(i)}}\right\}$, which form a specially chosen decomposition of $P$ as $P=\sum_i \ketbra{\psi_P^{(i)}}$.

Given a variable assignment to states $f:\mathcal{H}_d\rightarrow \{0,C,1\}$, we can define a variable assignment to all projection operators on $\mathcal{H}_d$, $g:\mathcal{P}(\mathcal{H}_d)\rightarrow \{0,C,1\}$ by the following:

\begin{equation}
g(P)=\begin{cases}
0\quad \mbox{if } \exists \left\{\psi_P^{(i)}\right\}, \>P=\sum_i \ketbra{\psi_P^{(i)}}, \>f\left(\ket{\psi_P^{(i)}}\right)=0\>\forall i,\\
1\quad \mbox{elif } \exists \left\{\psi_P^{(i)}\right\}, \>P=\sum_i \ketbra{\psi_P^{(i)}}, \>f\left(\ket{\psi_P^{(i)}}\right)\in\{0,1\}\>\forall i, \\
C\quad \mbox{o/w.}
\end{cases}
\end{equation}

The motivation for this notion of value inheritance is as follows: in any context containing $P$, we can consider replacing it with one of its maximally fine-grained decompositions, which then inherits value assignments from $f$ derived \emph{via} the annulus method. Since the assignment of values to this rank-1 decomposition of $P$ is necessarily consistent with the restrictions of the scenario, we can consider a post-processing that consists of a coarse-graining of those fine-grained results, ``forgetting'' which one of them occurred. This process must also be consistent. Hence, as long as there exists some decomposition of $P$ into rank-1 projectors, each of which is assigned a noncontextual valuation, then $P$ can inherit a noncontextual valuation. However if there must be a contextual vector in any such decomposition, we must treat $P$ as having a contextual valuation.

The challenge, then, is to characterise which projectors $P$ must be given a contextual valuation under this schema, and then prove a bound analogous to that of Theorem \ref{main}.

\begin{thm}
In a scenario in which the available measurements are rank-$r$ projectors $\{P_i\}$, peformed on a quantum system of dimension $d$, then at most $(I_{\nicefrac{1}{2}}(r,d-r)-I_{\nicefrac{r}{d}}(r,d-r))|\{P_i\}|$ are given a contextual valuation, where $I_x(\alpha,\beta)$ is the \emph{regularised incomplete beta function}.
\end{thm}
\begin{proof}
The proof proceeds similarly to the proof of Theorem \ref{main}. We will identify a generalised annulus within the space of rank-$r$ projectors on $\mathcal{H}_d$ which can be associated with contextual valuations, and then use its volume in proportion to that of the overall Hilbert space to form a bound on how many must be given contextual valuations.
\begin{lem}
Under the extension of a valuation of rank-1 projectors given by the annulus method from Theorem \ref{main} centred on $\ket{\phi}$, a rank-$r$ projector $P_r$ is contextual only if $\nicefrac{r}{d}\leq \tr(P_r\ketbra{\phi})\leq\nicefrac12$.
\end{lem}
\begin{proof}[Proof of Lemma]

First, consider the case that $\tr(P_r\ketbra{\phi})>\nicefrac12$. We could equivalently write this as $\bra{\phi}P_r\ket\phi=\bra{\phi}P_rP_r\ket\phi>\nicefrac12$. We note that $P_r\ket\phi$ is a vector in the  eigenvalue-1 subspace of $P_r$. As such, we can find a basis for $P_r$ that includes $P_r\ket\phi$, up to a normalisation constant. By design, each of these other basis elements $\ket{i}$ will have $\braket{i}{\phi}=\bra{i}P_r\ket{\phi}=0$. Hence, this corresponds to a decomposition of $P_r$ in which every vector would be given either a $1$ or a $0$ valuation, and so $P_r$ inherits a $1$ valuation.

Next, we consider the case that $\tr(P_r\ketbra{\phi})=\nicefrac{r}{d}-\epsilon<\nicefrac{r}{d}$. We wish to decompose $P_r$ in such a way that each of the elements $\ket{i}$ of the decomposition have $\sprod{i}{\phi}<\nicefrac1d$. Again, we consider the vector $P_r\ket\phi$ and choose a basis for the 1-eigenspace of $P_r$ so that we have
\begin{equation}
P_r\ket\phi=\sum_{i=0}^{r-1}\left(\sqrt{\frac1d-\frac{\epsilon}{r}}\right)\ket{i}.
\end{equation}
Each of these elements in the decomposition would recieve a noncontextual 0 valuation, so $P_r$ inherits a 0 valuation. This concludes the proof.
\end{proof}

Using the result of the lemma, we can apply a proof method identical to that in Theorem \ref{main} if we can calculate the proportion of the Hilbert space taken up by the set $\left\{P_r\middle|\nicefrac{r}{d}\leq \tr(P_r\ketbra{\phi})<\nicefrac12\right\}$.

As before, entries of a uniformly random unit vector in $\mathbb{C}^d$ have the same distribution as the entries of a column in a Haar-random $d\times d$ unitary matrix. Applying a result of \.{Z}yczkowski and Sommers \cite{Zycz2000}, if $U$ is a Haar-random $d\times d$ unitary, then defining $Y=\sqrt{\sum_{k=1}^r\left|U_{k,1}\right|^2}$, we have
\begin{equation}
P^Y_{d,r}(y)=c_{d,r}y^{2r-1}(1-y^2)^{d-r-1},
\end{equation}
where $P^Y_{d,r}(y)$ is the probability density function for the random variable $Y$, and
\begin{equation}
c_{d,r}=\frac{2}{B\left(r,d-r\right)}=\frac{2\Gamma(d)}{\Gamma(r)\Gamma(d-r)},
\end{equation}
in which $B$ is the Euler beta function. and $\Gamma$ is the gamma function. Performing a change of variables, then, in which $T=Y^2$, we get: 
\begin{align}
P^T_{d,r}(t)&=\left|\frac{\dif}{\dif t}\left(\sqrt{t}\right)\right|P^Y_{d,r}\left(\sqrt{t}\right)\\
&=\frac{t^{r-1}\left(1-t\right)^{d-r-1}}{B(r,d-r)}
\end{align}
This is exactly a Beta distribution with shape parameters $\alpha=r$ and $\beta=d-r$; we have derived that $T\sim\mbox{Beta}(r,d-r)$. We note taking $r=1$, and integrating, we recover our earlier result used in Theorem \ref{main}. Since the CDF for a Beta distributed random variable is given by the regularised incomplete beta function, $I_x(\alpha,\beta)$, the proportion of the Hilbert space taken up by the set $\left\{P_r\middle|\nicefrac{r}{d}\leq \tr(P_r\ketbra{\phi})\leq\nicefrac12\right\}$ is given by
\begin{equation}
I_{\nicefrac{1}{2}}(r,d-r)-I_{\nicefrac{r}{d}}(r,d-r).
\end{equation}
To complete the proof we once again apply the argument from Fubini's theorem used in the proof of Theorem \ref{main}.
\end{proof}
\begin{cor}
For any choice of $d$, $r$, the proportion of projectors requiring a contextual valuation is always less than \nicefrac12.
\end{cor}
\begin{proof}
We have for any $d$, $r$, that the proportion of projectors requiring a contextual valuation is bounded above by $I_{\nicefrac{1}{2}}(r,d-r)-I_{\nicefrac{r}{d}}(r,d-r)$. We wish, then, to bound this quantity above by $\nicefrac12$.

Since $I_x(\alpha,\beta)$ is the cumulative distribution function for a random variable with distribution $\mbox{Beta}(\alpha,\beta)$, a trivial upper bound for this quantity is given by
\begin{equation}
I_{\nicefrac{1}{2}}(r,d-r)-I_{\nicefrac{r}{d}}(r,d-r)< 1-I_{\nicefrac{r}{d}}(r,d-r).
\end{equation}
In fact for large $d$, this is a reasonable approximation, since we can apply Chebyshev's inequality to show that $I_{\nicefrac{1}{2}}(r,d-r)\geq1-O(d^{-2})$. We wish, then, to lower bound the value of $I_{\nicefrac{r}{d}}(r,d-r)$. The mean-median-mode inequality \cite{Kerm2011} for the Beta distribution with $\alpha\leq\beta$, which is met here since $r$ must be a nontrivial factor of $d$, is given by
\begin{equation}
\frac{\alpha-1}{\alpha+\beta-2}\leq m(\alpha,\beta)\leq\frac{\alpha}{\alpha+\beta},
\end{equation}
where $m(\alpha,\beta)$ represents the median of a random variable with a $\mbox{Beta}(\alpha,\beta)$ distribution. This is an equality only when $\alpha=\beta$; although this corresponds to a trivial contextuality scenario. Importantly, we have
\begin{equation}
m(r,d-r)\leq\frac{r}{d}.
\end{equation}
Hence by the monotonicity of $I_x(\alpha,\beta)$, we have $I_{\nicefrac{r}{d}}(r,d-r)\geq I_{m(r,d-r)}(r,d-r)=\nicefrac12$, and the result follows.

\end{proof}

\section{Future Work}
There is a large gulf between the best-known realisation of a quantum-mechanically accessible BKS proof and the bounds on the robustness of such a proof that we have derived in this article. However, it may be that to prove a tighter bound on what is quantum-mechanically possible in a way similar to that of this paper is mathematically difficult; this is because the method requires that we find a set that is independent, regardless of the specific structure of the set of vectors in question. This implies bounds on the independence number of the associated orthogonality graph for all finite contextuality scenarios, and by extension also the contextuality scenario that includes every quantum state. This problem, to choose the largest set of points on a hypersphere such that no two points in the set are orthogonal, is essentially a problem known as Witsenhausen's problem \cite{Wits1974}. For Witsenhausen's problem, it is conjectured \cite{DeCo2015} by Gil Kalai that the spherical-cap configuration used in this proof is in fact optimal, although this is not even known for the original three-real-dimensional case considered by Kochen and Specker. In turn, the solution to Witsenhausen's problem implies lower bounds for the Hadwiger-Nelson problem \cite{DeCo2015}: the number of colours needed to colour Euclidean spaces if no two points distance 1 away from each other can have the same colour. The solution to this problem is thought to depend on the specific model of ZF set theory adopted \cite{Soif2008}.

 In this paper, we considered contexts formed of projective matrices that sum to the identity. Recently, attention has been drawn to an analogous problem defined for sets of mutually commuting Hermitian matrices, and so an extension of variant of this result in that case would be helpful for bounding the amount of classical memory needed for an unbiased weak simulation of quantum subtheories \cite{Kara2017}. 

\section{Acknowledgements}
Special thanks to Mordecai Waegell, who suggested using the $E_8$ root system and two-qubit stabiliser quantum mechanics to demonstrate lower bounds on what values of the figure of merit in Section \ref{sec:main} were quantum-mechanically accessible, as well as his help calculating this lower bound.
Special thanks also to Terry Rudolph for asking this question in the first place.
I am grateful to Angela Xu, Angela Karanjai and Baskaran Sripathmanathan for their helpful discussions which have guided this project. I acknowledge support from EPSRC, \emph{via} the CDT in Controlled Quantum Dynamics at Imperial College London; and from Cambridge Quantum Computing Limited. This research was supported in part by Perimeter Institute for Theoretical Physics. Research at Perimeter Institute is supported by the Government of Canada through the Department of Innovation, Science and Economic Development and by the Province of Ontario through the Ministry of Research and Innovation.

\bibliography{/Users/andrewsimmons/Documents/Latex/Bib/bibliography}{}
\bibliographystyle{plain}

\end{document}